\documentclass[journal]{IEEEtran}

\usepackage{cite}
\usepackage{amsmath}
\usepackage{bm}
\usepackage{graphicx}
\usepackage{amsfonts}
\usepackage{amssymb}
\usepackage{hyperref}

\usepackage[capitalize]{cleveref}
\usepackage{comment}
\usepackage{multirow}
\usepackage[font=footnotesize, justification=raggedright]{caption}
\usepackage[font=footnotesize]{subcaption}
\usepackage[english]{babel}
\usepackage{color}
\newtheorem{theorem}{Theorem}
\begin{document}
%
\title{Performance of QUBO-Formulated MIMO Detection Under Hardware Precision Constraints}
%
%
%

\author{\IEEEauthorblockN{Seyedkhashayar~Hashemi, 
        \textit{Graduate~Student~Member}, \textit{IEEE}, 
        Elisabetta~Valiante,
        Ignacio~Rozada,
        and~Moslem~Noori}\\
\thanks{S. Hashemi and E. Valiante contributed equally to this work. Corresponding author: \textit{elisabetta.valiante@1qbit.com}.}%
\thanks{S. Hashemi was with 1QB Information Technologies (1QBit), Vancouver, Canada, and also with the Faculty of Engineering, University of Alberta, Canada. E. Valiante, I. Rozada, and M. Noori are with 1QBit, Vancouver, Canada.}
}

%


\maketitle

\begin{abstract}
The evolution of multiple-input, multiple-output (MIMO) systems requires the efficient detection algorithms to overcome the exponential computational complexity of optimal maximum likelihood detection. Reformulating MIMO detection as a quadratic unconstrained binary optimization (QUBO) problem enables the use of highly parallel, physics-inspired, hardware-accelerated solvers and non-von Neumann architectures. However, embedding continuous-valued QUBO coefficients into hardware introduces quantization noise due to finite precision, which can severely degrade detection accuracy. This paper presents a rigorous analysis of the performance impact of finite-precision, hardware-accelerated QUBO solvers in MIMO detection. We analytically derive the probability distribution functions of the QUBO matrix entries and introduce novel homogeneous and heterogeneous quantization schemes based on either instantaneous channel state information or its statistical features. We further derive a sufficient condition on the precision required to maintain the optimal solution after quantization. Extensive numerical experiments, across various MIMO system sizes and modulation orders (up to 256-QAM), show that heterogeneous quantization matches the full-precision baseline bit error rate using significantly fewer bits than homogeneous approaches. We provide hardware-aware guidelines for selecting the optimal quantization strategy.
\end{abstract}


%
\IEEEpeerreviewmaketitle

\section{Introduction}
\IEEEPARstart{T}{he} evolution of wireless communication systems in the last three decades has been driven by a growing demand for higher data rates and improved spectral efficiency. From the early days of single-antenna systems to the massive arrays deployed in a 5G network and envisioned for 6G, the underlying physical layer's technology has fundamentally shifted toward multiple-input, multiple-output (MIMO) systems as it offers significant improvement in spectral efficiency. This improvement comes at the price of increased computational complexity in solving the MIMO detection problem, that is, the problem of recovering transmitted data from a noisy linear combination of received signals at the receiver's antenna array.

The optimal detection strategy, called ``maximum likelihood detection'' (ML detection), minimizes the error rate of the detected signal by exhaustively searching the entire search space~\cite{Verdu1998}. However, the computational complexity of ML detection grows exponentially with the number of antennas and the modulation order, rendering it infeasible for practical, delay-sensitive systems. Consequently, research has focused on near-optimal, low-complexity detectors, and several suboptimal detection methods have been proposed over the years. For example, linear methods, such as zero-forcing (ZF), matched filter, and minimum mean square error (MMSE), offer low complexity~\cite{Paulraj2003} and close-to-optimal performance when the system is underloaded, which means that the number of receiver antennas is much larger than the number of transmitter antennas.
In a fully loaded system, where the number of transmitters is almost equal to the number of receiver antennas, the channel matrix becomes ill-conditioned, and the performance of the linear methods drops noticeably. In this case, nonlinear methods such as sphere decoding~(SD)~\cite{Viterbo1999}, successive interference cancellation~(SIC)~\cite{Foschini1996, Shlezinger2020}, tree search~\cite{Studer2010}, and lattice reduction (LR)~\cite{Wubben2004} can be used. These methods offer improved performance over the linear methods, at the cost of higher complexity. The complexity grows with the number of antennas, making the nonlinear detection methods computationally expensive for fully loaded massive MIMO systems. 

Recently, a paradigm shift has emerged to overcome the computational challenges of MIMO detection with the advent of physics-inspired computing, where the ML detection problem is converted to a combinatorial optimization problem solvable by Ising machines. By mapping the error minimization function to an Ising Hamiltonian or a quadratic unconstrained binary optimization (QUBO) problem, MIMO detection can be reformulated as the problem of finding the ground state of a physical system~\cite{Kim2019}. For higher-order modulation schemes (e.g., 16-QAM or 64-QAM), the problem naturally expands into a higher-order polynomial unconstrained binary optimization (PUBO) formulation~\cite{Kim2019,Zeng2025}. Although the PUBO formulation captures the complex variable interactions of the problem more naturally, most current hardware devices are restricted to pairwise (i.e., quadratic) interactions. 

Physics-inspired heuristics for MIMO detection have been implemented in traditional computing environments (i.e., on CPUs and GPUs)~\cite{Kim2021}, and various non-von Neumann architectures have been proposed to solve these formulations efficiently. These include quantum annealers (such as D-Wave Systems's devices)~\cite{Kim2019}, coherent Ising machines (CIM) based on optical parametric oscillators~\cite{Singh2022}, and \mbox{p-bit-based} detectors either simulated~\cite{Zeng2025} or implemented on \mbox{ASICs/FPGAs~\cite{Sajeeb2026, Zhu2026}.} These solvers offer the potential for massive parallelism and improved scaling, but they each face their own unique implementation hurdles regarding connectivity, analog noise, and, above all, the precision required to embed problems on them. 

In physical Ising machines, the problem coefficients are encoded in physical quantities, such as magnetic fields, optical amplitudes, or electrical conductances, that are inherently subject to noise, fabrication variability, and control limitations. The resolution with which these machines can distinctively encode energy differences determines their ``effective bit-width'' and limits the available precision to encode a problem's coefficients. For example, the precision of \mbox{D-Wave} Systems's hardware is a function of the entire control chain, from the room-temperature electronics to the millikelvin cryostat environment, and can be quantified to be approximately 4 to 5~bits~\cite{DWave}. In a CIM, the coupling coefficients are stored in the digital memory of the FPGA that performs feedback during the optimization process, so the coefficients' precision is bounded by the resolution and the memory bandwidth of the converters that interface the digital calculation with the analog optical network, typically reaching 6 to 14~bits~\cite{Honjo2021}. In memristor crossbar arrays, the coefficients are encoded directly in the physical conductance of a metal-oxide filament at the intersection of two wires, and the calculations are performed instantly via Kirchhoff's and Ohm's laws. The maximum achievable precision is about 11~bits, but more practical implementations operate with 4 to 8~bits of \mbox{precision~\cite{Shapero2013, Mahmoodi2018, Rao2023}.}  

Since the available precision on physical Ising machines is often very limited compared to that of traditional computers, the formulated QUBO problem has to be quantized to accommodate the limited precision. This introduces quantization noise, which can severely degrade detection performance if not properly managed.
To this end, it is important to answer two questions: i) Is such a limited precision sufficient to correctly perform MIMO detection by solving the quantized QUBO problem? ii) How can the effective precision of the Ising machine be increased to remedy the quantization noise? Although there have been previous studies to address the second question (see, for example, Ref.~\cite{Singh2024}), a response to the first question and a complete study of the effects of quantization on the QUBO formulation of the MIMO detection problem have not yet been addressed in the literature.   

In this paper, we present a rigorous and extensive analysis of the effect of quantization on the performance of QUBO solvers for MIMO detection. To this end, we first derive the statistical properties of the full-precision QUBO matrix coefficients. These properties are then used to propose statistical quantization schemes in addition to quantization schemes based on QUBO matrix realizations. We then analyze how quantization could affect the optimal solution of the QUBO problem and derive a sufficient condition to guarantee that the optimal solution of the full-precision QUBO problem is preserved in the quantized QUBO problem. To validate our analysis, we present the results of numerical experiments with a QUBO solver and evaluate how the quantization schemes affect the solver's performance. Finally, we propose some guidelines for choosing the best quantization scheme and precision to apply to save resources without degrading the solver's performance.  
 
\section{MIMO Detection: System Model and Problem Formulation}\label{sec: preliminaries}
\subsection{System Model}
We consider an uplink communication link between the $N_t$ transmit antennas and the $N_r$ receive antennas. The received complex-valued signal vector $\mathbf{y} \in \mathbb{C}^{N_r}$ at the receiver is
\begin{equation}
\label{received_sig}
    \mathbf{y} = \mathbf{Hx} + \mathbf{n},  
\end{equation}
where $\mathbf{H} \in \mathbb{C}^{N_r\times N_t}$ is the Rayleigh fading channel matrix between the transmitter(s) and the receiver, $\mathbf{x} \in \mathcal{S}^{N_t}$ is the vector of $N_t$ independently transmitted symbols with $\mathcal{S}$ being the set of all symbols in the modulation's constellation, and $\mathbf{n} \in \mathbb{C}^{N_r}$ is the vector of additive white gaussian noise (AWGN) with a variance of $\sigma_n^2$ and a power of $N_0$ on the receiver side. The quantities $\mathbf{y}$ and $\mathbf{H}$ are assumed to be known by the receiver. 

For modulation, we assume QAM with modulation order $M$; therefore, $\mathcal{S} \subseteq \mathbb{C}$ and $|\mathcal{S}|=M$. The number of bits transmitted by each symbol is $r=\log_2 M$. Here, we consider that $M$ is an even power of two and the QAM constellation has a standard square shape as shown in \cref{fig:M-QAM}. We set the minimum distance between two symbols in the constellation to be $d_{\rm min}=\sqrt{\frac{6}{M-1}}$ to normalize the average energy of the transmitted QAM signal to unity. This normalization ensures that all modulations transmit with the same power per symbol, allowing a fair comparison of performance, and that the energy per bit $E_{\rm b}$ is always $1/r$. The signal-to-noise ratio is defined as $10 \log (E_{\rm b}/N_0)$, and is measured in decibels.

\begin{figure}[t]
    \centering
    \includegraphics[width =1\linewidth, trim={7.5cm 5.5cm 13cm 6cm},clip]{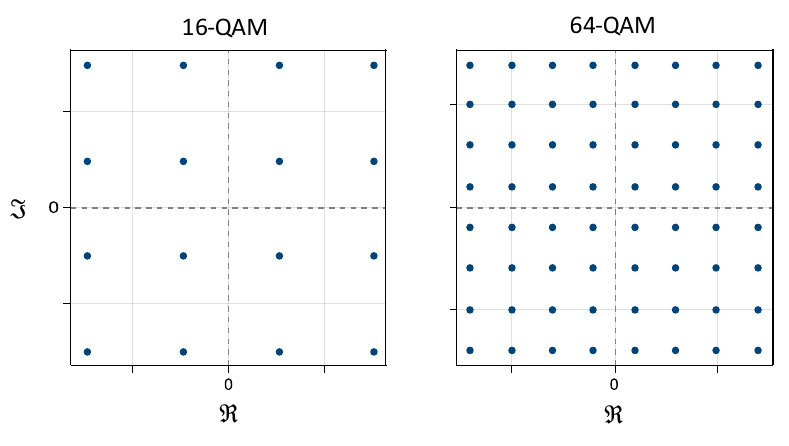}
    \caption{Examples of $M$-QAM constellations with $M=16$ (left) and $M=64$ (right). The minimum distance between symbols is $\sqrt{\frac{6}{M-1}}$, such that the constellation's average power is normalized to unity.}
    \label{fig:M-QAM}
\end{figure}

\subsection{Problem Formulation}
On the receiver side, the goal of MIMO detection is to determine $\mathbf{x}$ such that the objective function
\begin{equation}
\label{eq:MIMO_obj}
\arg \min_{\mathbf{x}\in \mathcal{S}^{N_t}} \| \mathbf{y} - \mathbf{Hx}\| ^2
\end{equation}
is minimized. The problem defined by expression~\eqref{eq:MIMO_obj} is referred to as the maximum likelihood (ML) detection problem, having a search space size of $\mathcal{S}^{N_t}$, which grows exponentially with the number of transmit antennas, making the ML problem computationally intractable. In the section that follows, we will discuss how the objective in expression~\eqref{eq:MIMO_obj} can be reformulated into a QUBO format. 

\section{MIMO Detection: QUBO Formulation} \label{sec: QUBO formulation}
A general QUBO problem can be defined by an upper triangular matrix $\mathbf{Q} \in \mathbb{R}^{N \times N}$ and formulated as
\begin{equation}
\label{eq:qubo_prob}
\min_{\mathbf{q}} \mathbf{q}^T \mathbf{Q} \ \mathbf{q} =   \sum_{i=1}^N\sum_{j\geq i}^N Q_{i,j}q_{i}q_{j},
\end{equation}
where $\mathbf{q} = [q_1, q_2,\dots, q_N]$, with $q_i \in \{0,1\}$, is the vector of binary variables of the QUBO problem and $Q_{i,j}$ is the element of $\mathbf{Q}$ located in the $i$-th row and the $j$-th column. Since $q_i^2$ is equal to $q_i$, when $i=j$, $Q_{i,j}q_{i}q_{j}$ can be written as $Q_{i,i}q_{i}$. These terms are associated with the diagonal entries of $\mathbf{Q}$ and are referred to as linear terms. The terms associated with the off-diagonal entries of $\mathbf{Q}$ are referred to as quadratic terms.

The formulation of the MIMO detection problem in expression~\eqref{eq:MIMO_obj} is a quadratic function of $\mathbf{x}$. Thus, if a linear mapping can be found that expresses $\mathbf{x}$ as a function of binary variables, the MIMO detection problem becomes a QUBO problem as a function of binary variables. To find such a linear mapping, let us consider the $i$-th element of $\mathbf{x}$, denoted by $x_i$. Defining a binary vector $\mathbf{q} \in \{0, 1 \}^{rN_t}$, we use $\{q_{(i-1)r+1}, \ldots, q_{ir}\}$, where $i \in \{1 , 2, \ldots, N_t\}$, to separately encode the in-phase (real) and quadrature (imaginary) components of $x_i$ as
\begin{equation} \label{eq:x_to_q}
x_i = T(q_{(i-\frac{1}{2})r+1},\dots,q_{ir}) + \textbf{j}\,T(q_{(i-1)r+1},\dots,q_{(i-\frac{1}{2})r}),
\end{equation}
where, for any $r/2$ arbitrary binary variables $\{ q'_1, q'_2, \ldots, q'_{\frac{r}{2}} \}$, $T(\cdot)$ is defined as
\begin{equation}
\begin{aligned}
\label{eq:linear_mapping}
&T(q'_1, q'_2, \ldots, q'_{\frac{r}{2}}) \\
&= \frac{d_{\min}}{2} \left[2(2^{\frac{r}{2}}q'_{\frac{r}{2}} + \ldots+2q'_2+q'_1) - (\sqrt{M}-1)\right].
\end{aligned}
\end{equation}
Note that $T(\cdot)$ allows for a simple linear mapping between symbols and variables, and it is different from Gray coding, which does not result in a linear mapping~\cite{Kim2019}.

Substituting each element of $\mathbf{x}$ in  expression~\eqref{eq:MIMO_obj} with its equivalent binary representation from \cref{eq:x_to_q},  we have
\begin{align} \label{eq:MIMO_to_QUBO} \nonumber
\Vert \mathbf{y} - \mathbf{Hx} \Vert ^2 &= (\mathbf{y}-\mathbf{Hx})^H(\mathbf{y}-\mathbf{Hx})\\ \nonumber 
&=\sum_{k=1}^{N_r}(y_k-(Hx)_k)^* (y_k-(Hx)_k)\\ \nonumber
&=\sum_{k=1}^{N_r} (c_k+W_{k,1}q_1+\ldots+W_{k,rN_t}q_{rN_t})^* \\ 
& \quad \cdot (c_k+W_{k,1}q_1+\ldots+W_{k,rN_t}q_{rN_t}),
\end{align}
where $(\cdot)^H$ denotes the conjugate transpose and $(\cdot)^*$ denotes the conjugate, $(Hx)_k$ refers to the $k$-th element of $\mathbf{H}\mathbf{x}$, 
\begin{align}
\label{eq:MIMO_to_QUBO_parameters} 
c_k&=y_k-\frac{d_{\min}}{2}\cdot(-(\sqrt{M}-1))\cdot(1+\textbf{j})\cdot\sum_{l=1}^{N_t}H_{k,l},
\end{align}
and
\begin{align}
\label{eq:MIMO_to_QUBO_parameters_2}
\nonumber
W_{k,m}\!&=\!-H_{k,\lfloor m/r \rfloor} \, e^{\, \textbf{j}\frac{\pi}{2} \lfloor\frac{(m \!\!\!\! \mod \! r)}{r/2} \rfloor} \\ 
& \quad \quad \quad \cdot 2^{[(m \!\!\!\! \mod \! r) \!\!\!\! \mod \! \frac{r}{2}]+1}\cdot \frac{d_{\min}}{2}.
\end{align}
As a result, the diagonal and off-diagonal entries of the QUBO matrix are
\begin{equation} \label{eq:QUBO_matrix_entries}
Q_{i,j} =
\begin{cases}
\displaystyle \sum_{k=1}^{N_r} \left( c_k^* W_{k,i} + c_k W_{k,i}^* + |W_{k,i}|^2 \right), & i = j, \\[10pt]
\displaystyle \sum_{k=1}^{N_r} \left( W_{k,i}^* W_{k,j} + W_{k,j}^* W_{k,i} \right), & i \ne j.
\end{cases}    
\end{equation}

\section{Probability Distributions of the QUBO Matrix's Entries}
\label{sec:qubo_distr}
In this section, we derive the probability distribution functions (PDF) of the entries of the QUBO matrix. These distributions will be used in the section that follows to analyze the performance of different quantization schemes.

We treat the diagonal (linear) and off-diagonal (quadratic) entries of the QUBO matrix separately. First, we present the following theorem for the diagonal entries.

\vspace{0.5 em}
\begin{theorem}\label{theorem:diagonal}
The probability distribution of the $i$-th diagonal entry of the QUBO matrix, $Q_{i,i}$, can be approximated by 
\begin{equation}
\begin{aligned}
\label{eq:QUBO_PDF_diag}
Q_{i,i}\sim &\ \mathcal{N}[\left(2\cdot N_r \cdot \Re(g_i)\right), \\ 
& \ \ (2 (|a_i|^2\cdot (N_t-1) \cdot 2\cdot \frac{2\sqrt{M}-1}{\sqrt{M}+1}\cdot N_r\\ 
& \ \ \ \ + N_r\cdot \sigma^2_n\cdot|a_i|^2 + N_r (f_i - |g_i|^2))],
\end{aligned}
\end{equation}
where $\mathcal{N}[\mu, \sigma]$ refers to a normal distribution with mean $\mu$ and standard deviation $\sigma$, and 
\begin{align}
    a_i & = -e^{\,\textbf{j}\frac{\pi}{2} \lfloor\frac{i \!\!\!\! \mod \! r}{r/2} \rfloor} \cdot 2^{[(i \!\!\!\! \mod \! r) \!\!\!\! \mod \! \frac{r}{2}]+1} \cdot \frac{d_{\min}}{2}, \\
    g_i & = \left(a_i\left[\sqrt{\frac{6}{M-1}}\cdot\frac{(\sqrt{M}-1)}{2} \cdot (1-\textbf{j})\right] +\frac{|a_i|^2}{2}\right)\!, \\ \nonumber
    f_i &=\left(|a_i|^2 \cdot 2\cdot \frac{2\sqrt{M}-1}{\sqrt{M}+1} + \frac{|a_i|^4}{4} \right) \\ 
    & \,\,\, + 4 \left(\Re\left[a_i\cdot\frac{|a_i|^2}{2}\cdot  \sqrt{\frac{6}{M-1}}\cdot\frac{(\sqrt{M}-1)}{2} \cdot (1-\textbf{j}) \right]\right)\!\!. 
\end{align}
\end{theorem}
\emph{Proof:} See Appendix~\ref{app:diagonal}.
\vspace{0.5em}

As is evident from Theorem~\ref{theorem:diagonal}, the distribution parameters depend on the parameters of the MIMO system, such as the modulation order $M$ and the number of transmit and receive antennas $N_t$ and $N_r$, respectively. For off-diagonal entries, we derive the following distributions.

\vspace{0.5em}
\begin{theorem}\label{theorem:off-diagonal}
The probability distribution of the off-diagonal entries, $Q_{i,j}$, of the QUBO matrix can be approximated by the following:
\begin{itemize}
    \item Case 1: $\lfloor \frac{i}{r} \rfloor \neq \lfloor \frac{j}{r} \rfloor$, 
    \begin{align}
    \label{eq:QUBO_PDF_case1}
          Q_{i,j}\sim &\ \mathcal{N}[0,\\  \nonumber
          & \,\,\,\, (2\cdot N_r\cdot|a_i|^2\cdot|a_j|^2)].
    \end{align}

    \item Case 2: $\lfloor \! \frac{i}{r} \rfloor  = \lfloor \frac{j}{r} \rfloor$ and $\lfloor \frac{(j \!\! \mod \! r)}{r/2} \rfloor \ne \lfloor\frac{(i \!\! \mod \! r)}{r/2} \rfloor$,
    \begin{equation}
    \label{eq:QUBO_PDF_case2}
        Q_{i,j} = 0.
    \end{equation}
    
    \item Case 3: $\lfloor \! \frac{i}{r} \rfloor  = \lfloor \frac{j}{r} \rfloor$ and $\lfloor \frac{(j \!\! \mod \! r)}{r/2} \rfloor = \lfloor\frac{(i \!\! \mod \! r)}{r/2} \rfloor$, 
    \begin{align}
    \label{eq:QUBO_PDF_case3}
        Q_{i,j}\sim &\  \mathcal{N}[(2\cdot N_r\cdot|a_i|\cdot|a_j|), \\ \nonumber
        & \,\,\,\, (4\cdot N_r\cdot|a_i|^2\cdot|a_j|^2)].
    \end{align}
\end{itemize}
\end{theorem}

\textit{Proof:} See Appendix~\ref{app:off-diagonal}. 
\vspace{0.5em}

In Theorem~\ref{theorem:off-diagonal}, case 1 pertains to the quadratic coefficients of the binary variables assigned to different transmit antennas, while cases 2 and 3 pertain to the quadratic coefficients of the binary variables for the same transmit antenna. The difference between case 2 and case 3 is that, in case 2, the variable $i$ is assigned to the in-phase component of the transmitted symbol and the variable $j$ is assigned to the quadrature component, while in case 3, both variables are assigned to either in-phase or quadrature components. 

To support the validity of Theorems~\ref{theorem:diagonal} and \ref{theorem:off-diagonal}, in Fig.~\ref{fig:QUBO_PDFs} we provide distributions obtained from numerical simulations for different values of $N_t$ and $N_r$. In this figure, we choose a representative entry of the QUBO matrix from the diagonal and different cases of non-diagonal entries. We choose $Q_{1,1}$, $Q_{1,r+1}$, $Q_{1,\frac{r}{2}+1}$, and $Q_{1,2}$ as representatives of the entries belonging respectively to the diagonal, the off-diagonal case 1, the off-diagonal case 2, and the off-diagonal case 3 distributions. For each plot, we use 1000 channel realizations of the system, assuming a Rayleigh fading model for $\mathbf{H}$ with signal-to-noise ratio $E_{\rm b}/N_0 = 20$ dB. As we can see, the theoretical PDFs closely match the empirical ones in all cases. Note that off-diagonal case 2 entries are not included in the figure: all generated samples are exactly zero, as is expected from \cref{eq:QUBO_PDF_case2}.

We note that these analytical results depend on the distribution of $\mathbf{H}$ and that assuming a different channel model would modify the distribution of the QUBO coefficients.

\begin{figure}[t]
    \centering
    \begin{subfigure}{0.24\textwidth}
        \centering
        \includegraphics[width=\linewidth]{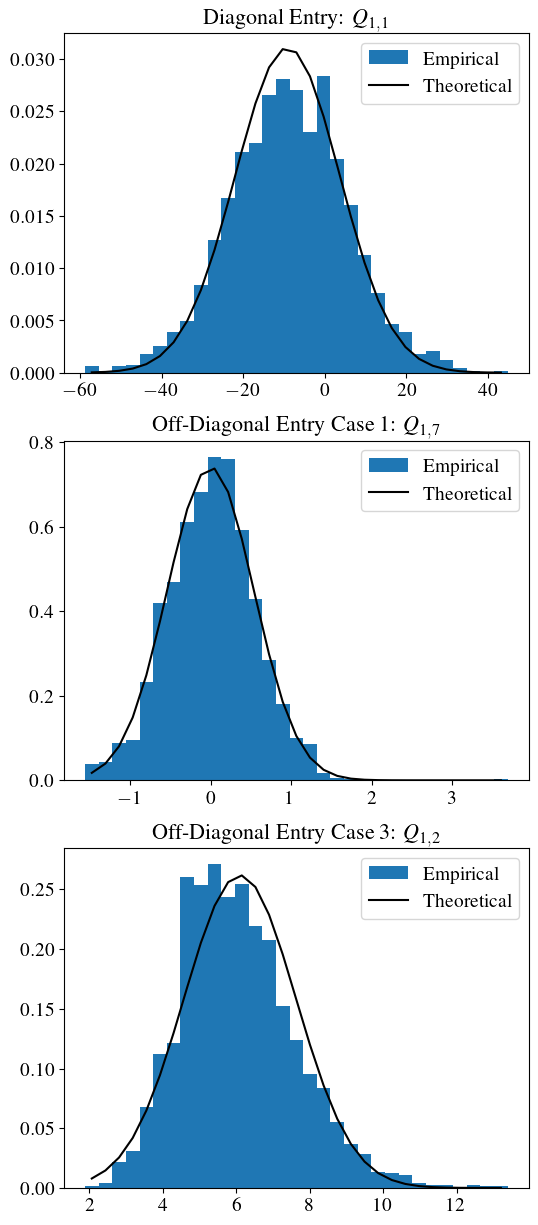}
        \caption{$N_t=N_r=16,\ M=64$}
        \label{NtNr16M64_PDF}
    \end{subfigure}
    \begin{subfigure}{0.24\textwidth}
        \centering
        \includegraphics[width=\linewidth]{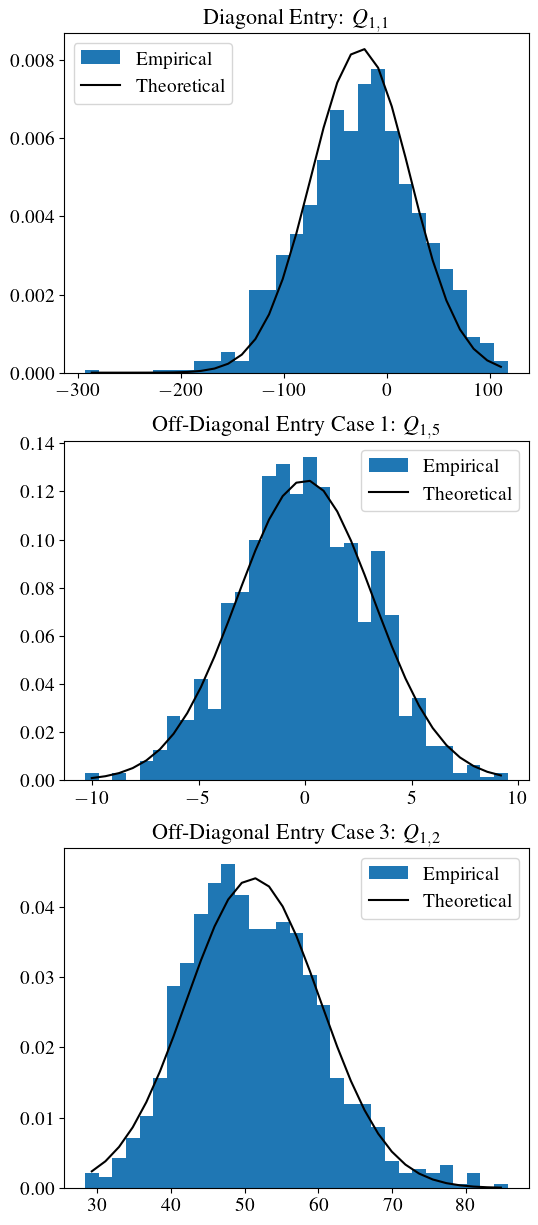}
        \caption{$N_r=N_t=32,\ M=16$}
        \label{NtNr32M16_PDF}
    \end{subfigure}
    \caption{Probability distribution functions (PDF) of QUBO matrix entries for systems with  the parameters \mbox{$N_t=N_r=16$,} $M=64$ (left) and $N_t=N_r=32$, $M =16$ (right). The representative matrix entries have indices $(i,j)$ equal to $(1,1)$ for the diagonal entries, $(1, r+1)$ for the off-diagonal case 1 entries, and $(1,2)$ for the off-diagonal case 3 entries. Theoretical distributions are reported in the PDFs \eqref{eq:QUBO_PDF_diag}, \eqref{eq:QUBO_PDF_case1}, and \eqref{eq:QUBO_PDF_case3}, for diagonal, case 1, and case 3 entries, respectively. Empirical distributions are created with 1000 realizations of the QUBO matrix $\mathbf{Q}$ for each value of $N_t$, $N_r$, and $M$. We also generate entries with indices $(1, \frac{r}{2}+1)$ for the off-diagonal case 2: as expected from \cref{eq:QUBO_PDF_case2}, they are all equal to 0, so they are not included in the figure.}
    \label{fig:QUBO_PDFs}
\end{figure}

\section{Mapping the QUBO Matrix to Hardware Constraints}
\label{sec:QUBO_Preproc}

Once \cref{eq:MIMO_obj} has been formulated as a QUBO problem, some preprocessing steps are performed on the matrix before attempting to find a solution with the QUBO solver. These preprocessing steps include normalization and quantization. Normalization is conducted so that the matrix entries fit the range of coefficients acceptable by the hardware (typically $[-1,1]$). Quantization is performed due to the limited precision of the hardware in storing the QUBO coefficients. 

\subsection{Quantization}
We consider a uniform quantization where the range of the input signal is divided into equally spaced intervals. The number of bits of precision $n_{\rm b}$ that we assign to the hardware determines the number of intervals $2^{n_{\rm b}}$. The \emph{range} and \emph{${n_{\rm b}}$} are the parameters of our quantization process.

Depending on the hardware's capabilities, two strategies can be considered for quantization. One strategy, which we call ``homogeneous'', is to quantize all entries in the QUBO matrix using the same precision. The second strategy, which we call ``heterogeneous'', is to use two different precision levels for the diagonal and off-diagonal entries of the matrix, specifically, to keep the diagonal terms at high precision and quantize only the off-diagonal terms. The justification for a heterogeneous approach is the existence of hardware~\cite{Mahmoodi2019, Cai2020, Hizzani2024, Bhattacharya2024} where the diagonal terms, which correspond to the linear terms of the problem, can be embedded with high precision, while the available precision for the quadratic terms is limited. 

\subsection{Preprocessing Parameters}
Certain parameters must be determined in order to perform the preprocessing steps. For normalization, a \emph{scaling factor} needs to be defined. For quantization, the number of precision bits \emph{$n_{\rm b}$} and the quantization \emph{range} must be set. While $n_{\rm b}$ depends on the needs of the user and the capabilities of the hardware, the other two parameters depend on the problem, and we use two approaches to set them: a \emph{per-realization} approach, in which the parameters are determined based on each realization of the channel matrix $\mathbf{H}$, and consequently the QUBO matrix $\mathbf{Q}$; and a \emph{statistical} approach, in which the parameters are set based on the features of the probability distributions derived in Section~\ref{sec:qubo_distr}.   
The statistical approach is faster compared to the per-realization approach, since the preprocessing parameters can be calculated a single time offline. This makes the statistical approach more resource efficient, as it requires fewer hardware components compared to the per-realization approach. However, this could come at the cost of degraded performance compared to the per-realization approach as it is less responsive to changes in the channel condition. For both approaches, we discuss how normalization is conducted and propose several quantization schemes. 

\subsection{Per-Realization Mapping}
\subsubsection{Normalization}
\label{sec:norm_per_realiz}
Normalization is performed by dividing all entries by the maximum absolute value of the QUBO matrix $\mathbf{Q}$. Therefore, for each realization, the scaling factor is the reciprocal of the maximum absolute value of the matrix. This ensures that all normalized entries are within the range $[-1,1]$.
\subsubsection{Quantization Schemes}
\label{sec:quant_per_realiz}
We define three quantization schemes to be applied after normalization.
\begin{itemize}
    \item \textbf{Homogeneous}:
    All entries in the matrix are quantized using $n_{\rm b}$ bits of  precision within the range $[-1,1]$.
    \item \textbf{Large Off-Diagonal}:
    The diagonal entries are retained at the original precision, while all off-diagonal coefficients are uniformly quantized over the full dynamic range of the off-diagonal matrix entries after normalization. The quantization will be heterogeneous.    
    \item \textbf{Small Off-Diagonal}:
    This is similar to the large off-diagonal scheme except that, for the off-diagonal entries, rather than assigning the quantization levels based on their full dynamic range, we choose the quantization levels based on the dynamic range of entries in cases 1 and 2 of Theorem~\ref{theorem:off-diagonal}. We call this the small range. Case 3 matrix entries that fall outside the small range are clipped, that is, their value will be set equal to the minimum or the maximum of the small range and then quantized. The motivation for using the small range rather than the full dynamic range is to provide more precision for the majority of the off-diagonal terms.    
\end{itemize}

\subsection{Statistical Mapping}
\subsubsection{Normalization}
\label{sec:norm_stat}
Here, the scaling factor is determined statistically. That is, instead of using the maximum values of the entries of $\mathbf{Q}$, we use \emph{statistical} maximum values derived from their distribution. Specifically, the statistical maximum absolute value for each coefficient is calculated as the sum of its mean and five times its standard deviation. Then, the largest value across the statistical maximum values of all coefficients is chosen for the scaling factor. It is possible that after applying the statistical scaling factor to the QUBO matrix, some values fall outside the range $[-1,1]$. Therefore, in the end, a clipping is required, where values smaller than $-1$ or larger than $1$ are projected to $-1$ and $1$, respectively.

\subsubsection{Quantization Schemes}
\label{sec:quant_stat}
This is the statistical counterpart to the three per-realization quantization schemes described above in \cref{sec:quant_per_realiz}.
\begin{itemize}
    \item \textbf{Homogeneous}:
    Uniform quantization is performed on the statistically normalized entries of $\mathbf{Q}$.
    \item \textbf{Large Off-Diagonal}:
    The difference between this scheme and its per-realization counterpart is that the dynamic range is calculated statistically. We consider the distributions of the entries of case 1 and case 3 (the entries of case 2 are 0) in PDFs \eqref{eq:QUBO_PDF_case1} to \eqref{eq:QUBO_PDF_case3}, and find the smallest statistical minimum and the largest statistical maximum in these entries to determine the dynamic range. Any off-diagonal entry outside the dynamic range will be clipped.
    \item \textbf{Small Off-Diagonal}:
    This scheme is similar to its per-realization counterpart, except that the dynamic range is determined by the statistical minimum and maximum over the entries in case 1 and case 2.
\end{itemize}

\section{Effect of Mapping on the QUBO Problem's Optimal Solution}
\label{sec:preprocess_effect}
The preprocessing steps we have discussed in \cref{sec:QUBO_Preproc} induce quantization and clipping errors on the entries of the QUBO matrix. These errors can accumulate in such a way that the optimal solutions of the quantized and non-quantized QUBO problems are not the same. As a result, even an ideal QUBO solver is unable to find the configuration that corresponds to the optimal solution of the non-quantized problem, and a detection error occurs. Thus, it is important to identify the criteria, for example, the required precision, under which the optimal solution is preserved after preprocessing. Such criteria are crucial in deciding on the applicability of a QUBO solver to the MIMO detection problem formulated as a quantized QUBO problem. 

Let us denote the QUBO matrix after normalization by $\mathbf{Q}$ and the QUBO matrix after normalization and quantization by $\mathbf{\widehat{Q}}$.
The optimal solution to the QUBO problem identified by $\mathbf{Q}$ is identical to the optimal solution of the original problem, since optimality is preserved when scaling by a fixed value. Now, given a configuration $\mathbf{q}^{(i)}\in \{0, 1\}^{rN_t}$ for the binary variables, the value of its objective function, or energy, is
\begin{equation}
\label{energy_eq}
E(\mathbf{q}^{(i)})= (\mathbf{q}^{(i)})^{T}\mathbf{Q}\, \mathbf{q}^{(i)},
\end{equation}
where $(\cdot)^T$ denotes the transpose. Suppose $\mathbf{q}^{\text{opt}}$ is the optimal configuration (i.e., the configuration with the lowest energy) for the QUBO problem defined by $\mathbf{Q}$.  Our goal is to find conditions where $\mathbf{q}^{\text{opt}}$ is also the optimal solution to the QUBO problem defined by $\mathbf{\widehat{Q}}$ to ensure that clipping and quantization do not introduce errors. 

Let us define the energy difference $\epsilon\ge0$ between the optimal configuration  $\mathbf{q}^{\text{opt}}$ and the second-best configuration (i.e., the configuration with the second-lowest energy) $\mathbf{q}^{\dagger}$ as
\begin{equation}
\label{gap_eq}
\epsilon=E(\mathbf{q}^{\dagger})-E(\mathbf{q}^{\text{opt}}).
\end{equation}

These two binary configurations are associated with two transmit vectors $\mathbf{x^{\text{opt}}}$ and $\mathbf{x^{\dagger}}$ in the $\mathcal{S}^{N_t}$ space. For sufficiently large signal-to-noise ratios, it is safe to assume that the second-best solution of the expression~\eqref{eq:MIMO_obj}, $\mathbf{x^\dagger}$, differs from the optimal solution $\mathbf{x}^{\text{opt}}$ only in the $j^*$-th entry, where \mbox{$j^* = \arg\min_{j} \Vert \mathbf{H}_j \Vert ^2$} and $\mathbf{H}_j$ is the set of entries of the $j$-th column of $\mathbf{H}$. In addition, the differing entries are two neighbouring points of the constellation, that is, $||\mathbf{x^{\dagger}} - \mathbf{x^{\rm opt}}|| = d_{\min}$. As a result, \cref{gap_eq} can be rewritten as
\begin{equation}
\begin{aligned}
    \label{eq:epsilon}
\epsilon &=E(\mathbf{q}^{\dagger})-E(\mathbf{q}^{\text{opt}})=\alpha\left(\| \mathbf{y} - \mathbf{Hx^{\dagger}}\| ^2 - \| \mathbf{y} - \mathbf{Hx^{\text{opt}}}\| ^2\right) \\
& = \alpha\left(\| \mathbf{H}(\mathbf{x^{\text{opt}}}-\mathbf{x^{\dagger}})+\mathbf{n}\|^2- \| \mathbf{n}\|^2\right)  \\
& = \alpha \left(d_{\min}^2 \Vert \mathbf{H}_{j^*} \Vert ^2 + 2 d_{\min} \Re(\mathbf{H}_{j^*} \mathbf{n}^H) \right),
\end{aligned}
\end{equation}
where $\alpha$ is the inverse of the scaling factor used to normalize the QUBO problem.

Now, let us define the maximum energy error induced by quantization for a given channel realization and $\mathbf{Q}$ as 
\begin{equation}
\label{eq:delta}
\delta = \max_{i\in \{1,2,...,2^{rN_t}\}} |E(\mathbf{q}^{(i)})-\widehat{E}(\mathbf{q}^{(i)})|,
\end{equation}
where $\widehat{E}(\mathbf{q}^{(i)})= (\mathbf{q}^{(i)})^{T}{\mathbf{\widehat{Q}}}\, \mathbf{q}^{(i)}$. 
It is guaranteed that the optimal solution will be preserved after quantizing $\mathbf{Q}$ if 
\begin{equation}
\label{eq:delta_min_epsilon}
\delta <\frac{\epsilon}{2}.
\end{equation}
To derive inequality~\eqref{eq:delta_min_epsilon}, the worst-case scenario is considered, in which the energy of the optimal solution increases by $\delta$ after quantization, $\widehat{E}(\mathbf{q}^{\text{opt}}) - E(\mathbf{q}^{\text{opt}}) = \delta$, while the energy of the second-best solution decreases by $\delta$,  $E(\mathbf{q}^{\dagger}) - \widehat{E}(\mathbf{q}^{\dagger}) = \delta$. In this case, $\widehat{E}(\mathbf{q}^{\text{opt}})$ is still the global minimum of $\mathbf{\widehat{Q}}$ only if the condition in inequality~\eqref{eq:delta_min_epsilon} is satisfied.  
Note that the condition in inequality~\eqref{eq:delta_min_epsilon} is sufficient but not necessary. For example, if quantization affects the solutions so that, for all $i$, $E(\mathbf{q}^{(i)})-\widehat{E}(\mathbf{q}^{(i)}) = \delta$, the quantized problem would still have the same optimal solution as the original problem for any value of $\delta$.

Furthermore, enforcing the condition in inequality~\eqref{eq:delta_min_epsilon} for every channel realization matrix $\mathbf{H}$ and the QUBO matrix $\mathbf{Q}$ is too conservative and  requires a very high degree of precision, most likely higher than the precision typically available in QUBO solvers. In addition, obtaining an analytical form to enforce such a constraint across all matrices $\mathbf{Q}$ is very complex and outside the scope of this work. To this end, we first define  $\delta_{\rm up}$ as 
\begin{equation}
\begin{aligned}
\label{eq:upper_delta}
\delta_{\rm up} = &  \sum_{i=1}^{rN_t}\sum_{j= i}^{rN_t}|Q_{i,j}-\widehat{Q}_{i,j}| \\
 & \geq \max_{i\in \{1,2,...,2^{rN_t}\}} |E(\mathbf{q}^{(i)})-\widehat{E}(\mathbf{q}^{(i)})|,
\end{aligned}
\end{equation}
where the right-hand side of the inequality is equal to $\delta$, as in \cref{eq:delta}.

Now, we relax the condition in inequality~\eqref{eq:delta_min_epsilon} and instead use 
\begin{equation}
    \label{eq:cond_general}
    \delta_{\max}<\frac{\epsilon_{\min}}{2}.
\end{equation}

To find $\epsilon_{\min}$, we first apply the central limit theorem to \cref{eq:epsilon}, to find the distribution of $\epsilon$, assuming that $N_r$ is sufficiently large. We derive the following expression by noticing that both the expected value and the variance of $|H_{i,j}|^2$ are $1$, while $\Re(H_{i,j} n_i)$ has an expected value of $0$ and a variance of $\sigma_n^2/2$: 
\begin{equation}
\begin{aligned}
    \label{eq:epsilon_PDF}
\epsilon& \sim \mathcal{N}[(\alpha\cdot d_{\min}^2\cdot N_r),\ (\alpha^2\cdot d_{\min}^4\cdot N_r)] \\
&\ \ \ + \mathcal{N}[0,\ (2\cdot \alpha^2 \cdot d_{\min}^2 \cdot N_r \cdot \sigma_{n}^2)] \\
& \sim\mathcal{N}[(\alpha \cdot d_{\min}^2\cdot N_r),\ ((d_{\min}^4+2\cdot d_{\min}^2\cdot \sigma_{n}^2)\cdot \alpha^2 \cdot N_r)].
\end{aligned}
\end{equation}
Then, we define $\epsilon_{\min}$ as the $p$-th percentile of the above normal distribution.  For $\delta_{\max}$, we take a numerical approach by first calculating $\delta_{\rm up}$ from \cref{eq:upper_delta} on a large number of channel realizations and then setting $\delta_{\max}$ to the $p$-th percentile of the collected samples. 

In what follows, we compare the analytical results with the numerical simulations to evaluate the performance of the proposed quantization schemes and validate our analysis. \Cref{fig:delta_max_th} presents the relationship between $\delta_{\max}$ and the number of bits used in quantization for different quantization schemes, and shows the precision required to meet the condition in inequality~\eqref{eq:cond_general}. To calculate $\delta_{\max}$ and $\epsilon_{\min}$ as explained above, we set $p=99$ and use 5000 realizations of MIMO systems with a signal-to-noise ratio of $E_{\rm b}/N_0 = 10$~dB. We observe that the number of precision bits required for $\delta_{\max}$ to be below the $\epsilon_{\min}/2$ threshold depends on the mapping method and increases with the size of the MIMO system and modulation. Large off-diagonal methods require lower precision than homogeneous methods to have their $\delta_{\max}$ values below the threshold. For both of these schemes, the per-realization and statistical methods perform very similarly. Large and small off-diagonal per-realization methods perform very similarly in low modulation, while small off-diagonal statistical  methods are always above the threshold. At higher modulations, both small off-diagonal methods are always above the threshold.

\begin{figure}[t]
    \centering
    \includegraphics[width=\linewidth]{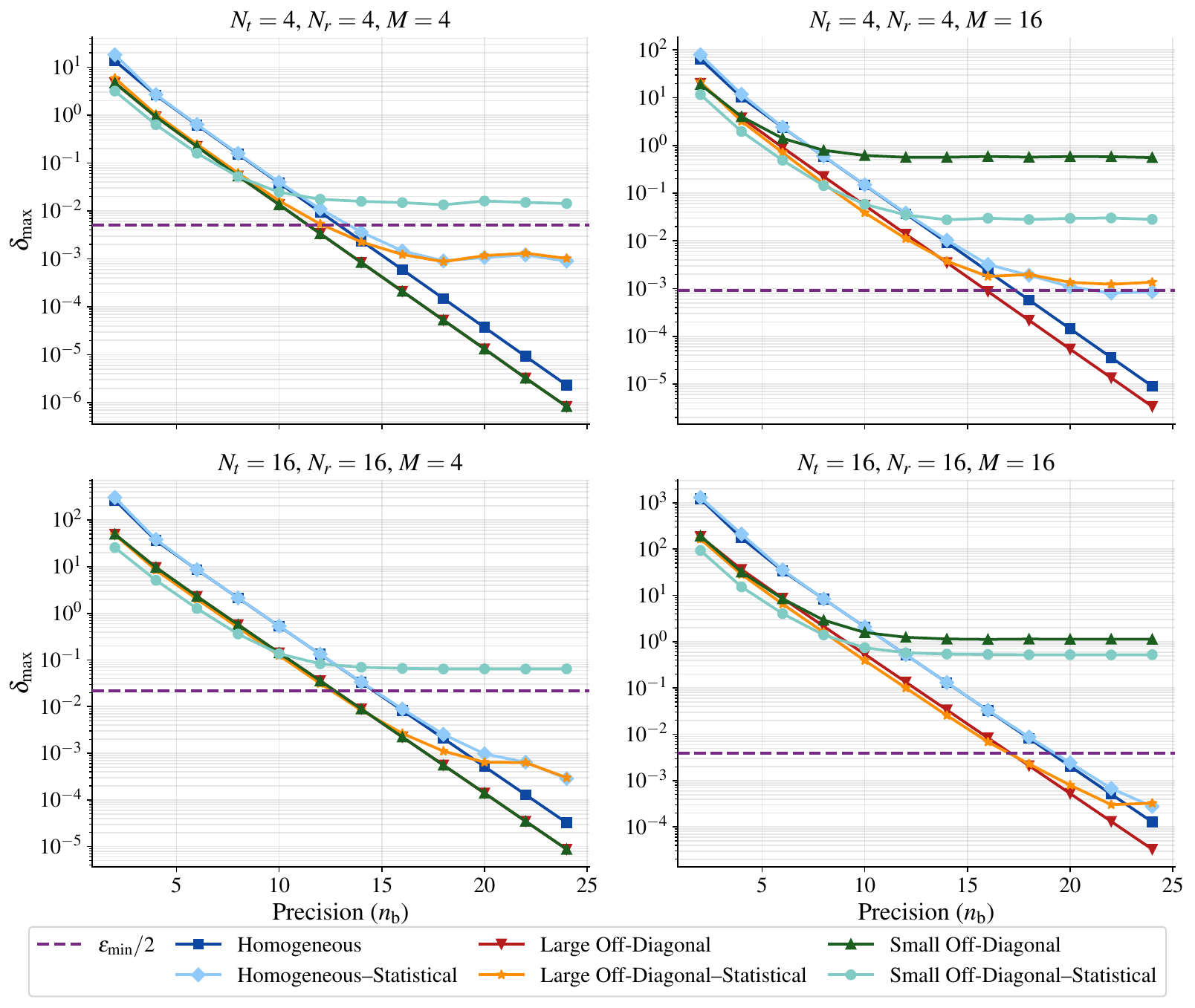}
    \caption{Maximum quantization error limit $\delta_{\max}$ versus precision for different MIMO systems with $E_{\rm b}/N_0=10$~dB.}
    \label{fig:delta_max_th}
\end{figure}

The required number of precision bits stated above is expected to be very conservative because the condition in inequality~\eqref{eq:cond_general} is not necessary, but sufficient, to preserve the optimal solution. Nevertheless, we believe that comparing the relative performance of different methods can be considered valid. In general, the large off-diagonal mapping schemes seem to be the best choices, as their values of  $\delta_{\max}$ are below the $\epsilon_{\min}/2$ threshold at a lower precision than the other methods. We confirm this claim with extended experiments in the section that follows.

\section{Experimental Results}\label{sec:experiments}
\renewcommand{\arraystretch}{0.85}
\begin{table*}[t]
    \centering
    \begin{tabular}{||c|c|c|c|c||}
        \hline
        $\bm{N_t\times N_r}$ & $\bm{M}$ & $\bm{n_{\rm b}}$ & $\bm{E_{\rm b}/N_0}$ &\textbf{Realizations} \\
        \hline\hline
        \multirow{4}{*}{$4\times 4$} & $4$ & \multirow{4}{*}{$2, 4, 6, 8, 10, 12, 14, 16$} & \multirow{4}{*}{$5, 10, 15$} &\multirow{2}{*}{10,000} \\
                    & $16$ & & &  \\
                    & $64$ & & & \multirow{2}{*}{1000} \\
                    & $256$ & & & \\
        \hline
        \multirow{4}{*}{$16\times 16$} & $4$ & \multirow{2}{*}{$2, 4, 6, 8, 10, 12, 14, 16$} & \multirow{2}{*}{$5, 10, 15$} &\multirow{2}{*}{10,000} \\
                    & $16$ & & &  \\
                    & $64$ & \multirow{2}{*}{$2, 4, 6, 8, 10, 12, 14, 16, 18, 20, 22, 24$} & \multirow{2}{*}{$5, 10, 15, 20, 25, 30$} & \multirow{2}{*}{1000} \\
                    & $256$ & & & \\
        \hline
        \multirow{4}{*}{$32\times 32$} & $4$ & \multirow{2}{*}{$2, 4, 6, 8, 10, 12, 14, 16$} & \multirow{2}{*}{$5, 10, 15$} &\multirow{2}{*}{10,000} \\
                    & $16$ & & &  \\
                    & $64$ & \multirow{2}{*}{$2, 4, 6, 8, 10, 12, 14, 16, 18, 20, 22, 24$} &
                    \multirow{2}{*}{$5, 10, 15, 20, 25, 30$} & \multirow{2}{*}{1000} \\
                    & $256$ & & & \\
    \hline
    \end{tabular}
    \caption{Configurations of all MIMO detection experiments. We analyzed three systems of different size, each with the same four modulations. For each system, we generated 10,000 realizations for the modulations $M=4$ and $16$, and 1000 realizations for $M=64$ and $256$. We explored a range of signal-to-noise ratios and numbers of bits of precision: from 5 to 15~dB and 2 to 16 bits for all modulations of the smallest system and for $M=4$ and $16$ of the larger two systems; and from 5 to 30~dB and 2 to 24 bits for $M=64$ and $256$ of the larger two systems. 
    }
    \label{tab:exp_configurations}
\end{table*}

\begin{figure*}[t]
    \centering
    \begin{subfigure}{0.48\linewidth}
        \includegraphics[width=\linewidth]{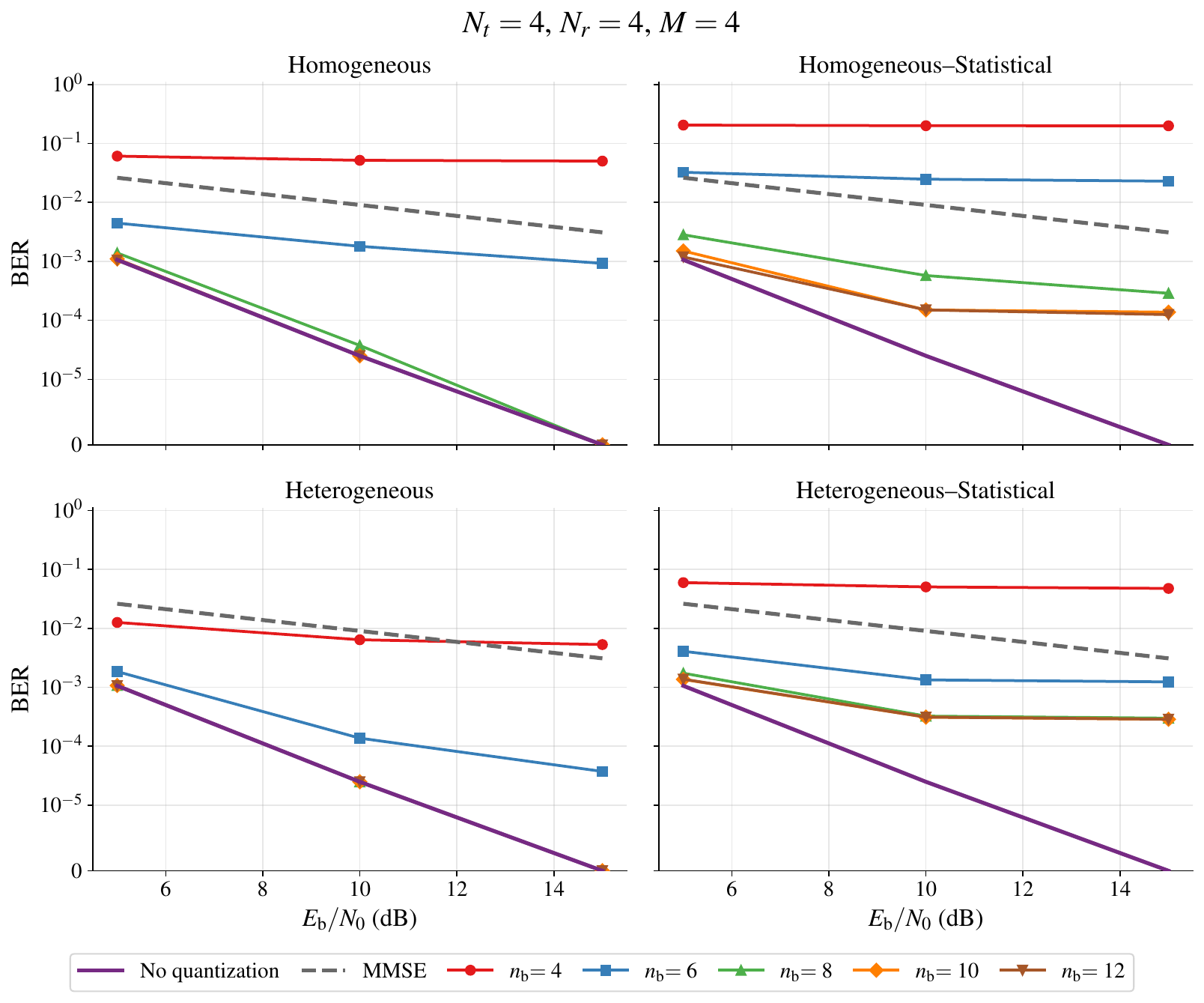}
    \end{subfigure}
    \begin{subfigure}{0.48\linewidth}
        \includegraphics[width=\linewidth]{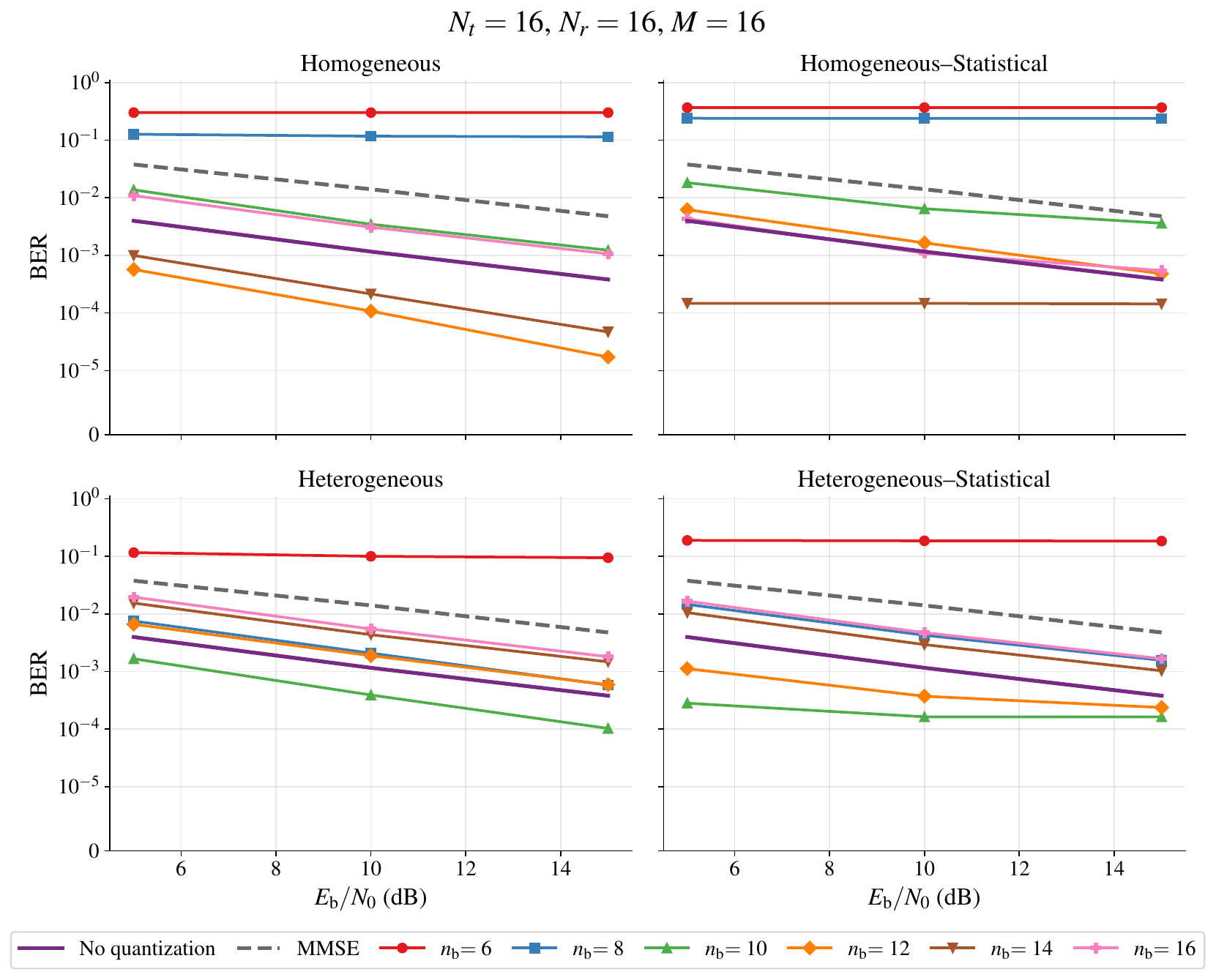}
    \end{subfigure}
    \caption{Bit error rate for MIMO detection performed after applying preprocessing with different quantization schemes and precision.}
    \label{fig:BER_bench}
\end{figure*}

In this section, we discuss our systematic exploration of the impact of preprocessing and the number of precision bits on the detection accuracy in several MIMO systems. A summary of all the experiment configurations is presented in Table~\ref{tab:exp_configurations}.

Each QUBO problem instance was solved using a parallel tempering solver (\texttt{PySA}~\cite{pysa}) using the result configuration from the MMSE method to initialize the algorithm. The temperature schedule used hot and cold temperatures computed according to the statistical properties of the QUBO terms~\cite{pysa}. To ensure a fair comparison across the system sizes and modulation orders, we fixed the total computational budget to 200,000 Monte Carlo steps each with an individual spin-flip, distributed proportionally according to the size of the QUBO problem $N=r\times N_t$.

The QUBO solver was configured with fixed parameters: \texttt{num\_replicas} $= 5$, \texttt{num\_reads} $=20$ (this parameter corresponds to the number of algorithm restarts), and \texttt{num\_sweeps} $=\,\, $200,000$/N$. The parameters were not tuned in advance, since the goal of this study has been to quantify the degradation of solvers' performance due to quantization and to define the best quantization scheme with the closest results to the full-precision implementation, and not to define the absolute best performance the solver can achieve or to compare it with other solvers. All experiments utilized single-precision floating-point arithmetic (\texttt{float32}) for computational efficiency, as preliminary tests showed negligible accuracy differences compared to double precision while providing a speed-up of approximately $20\%$. The binary solution returned by the QUBO solver was mapped back to the symbols of the constellation and compared to the original bit-string generated for each realization \cite{Kim2019}. The bit error rate (BER) was measured as the ratio of the number of erroneous bits divided by the total number of transmitted bits.  

\Cref{fig:BER_bench} shows BER results for a subsample of the experimental settings in Table~\ref{tab:exp_configurations} for brevity and clarity. A summary of all of our results is presented in \cref{fig:n_bits_all}. The small off-diagonal schemes consistently degraded the solver's performance regardless of precision and are omitted: from this point onward, we refer to the large off-diagonal scheme simply as ``heterogeneous''.

Four results stand out in all systems. One such result is that the full-precision baseline BER always outperforms MMSE, confirming the competitiveness of the QUBO formulation for MIMO detection~\cite{Kim2019, Kim2021, Singh2022, Zeng2025}, with improvements achievable at as few as 6 bits for the smallest systems. Another result is that both homogeneous and heterogeneous schemes approach the baseline BER when sufficient precision is used, with heterogeneous schemes performing better for larger systems, consistent with \cref{fig:delta_max_th} and validating the analysis in \cref{sec:preprocess_effect}. Third, the precision required increases alongside modulation order $M$, as expected due to the resulting increase in the size of the QUBO matrix and the decrease in the value of $d_{\rm min}$. Fourth, quantization occasionally improves the solver's performance relative to full precision---a counterintuitive but known effect for heuristic solvers, where quantization noise can help in escaping local minima~\cite{Seok2022}, and where the induced coefficient rounding can incidentally bring the temperature schedule closer to its optimal values. We confirmed this effect by testing several temperature schedules, and noticed that some schedules improved the solver's performance in solving the full-precision problem. Determining the optimal temperature schedule for each MIMO realization is highly impractical, as it would substantially decrease the throughput of the detector. 

The summary in \cref{fig:n_bits_all} confirms that heterogeneous methods achieve a better BER at lower precisions than homogeneous methods, consistent with the analysis in \cref{sec:preprocess_effect}. However, one discrepancy arises: while \cref{fig:delta_max_th} predicts a minimal difference between per-realization and statistical methods, experiments show that sometimes the per-realization methods perform better. This is because \cref{eq:delta} accounts only for the quantization error, while the experiments also introduce a clipping error at the statistical normalization stage (see \cref{sec:norm_stat}) that was not captured in the theoretical analysis.

Finally, we measured the effect of quantization independently of the solver. We used an exhaustive solver to evaluate the optimal configurations of the original problem and quantized problems using 1000 realizations of MIMO systems with a signal-to-noise ratio $E_{\rm b}/N_0 = 10$~dB. These results are presented only for MIMO systems of size $4 \times 4$, as running an exhaustive search for larger systems is not practical. \Cref{fig:ber_exhaustive} shows how, for the system where $M=4$, we can achieve a performance improvement over that of MMSE and match the full-precision BER results with as few as 6 and 8~bits of precision, respectively. At higher modulations, at least 8 and 14 bits of precision are required. These results are consistent with the \texttt{PySA} solver's results in \cref{fig:n_bits_all} and with the discussion in Ref.~\cite{Zhu2026}.

\section{Discussion and Conclusions}
The choice of the best quantization scheme depends on the flexibility and speed of the hardware on which the MIMO detection was implemented. According to our experimental results, the scheme requiring the fewest precision bits to achieve the lowest BER is the heterogeneous scheme. However, in adopting this scheme, some considerations should be taken into account. First, this is a per-realization scheme, which means that the normalization factor and the numerical values defining the quantization range must be obtained for each channel realization. This would require an extra calculation step and could potentially reduce the throughput of the system. The choice between a per-realization or statistical approach in preprocessing must be evaluated in the context of being a trade-off between the computational time and the number of precision bits required to achieve the lowest BER. Second, this scheme assumes that the linear (diagonal) terms of the QUBO problem can be mapped onto the hardware with a high level of precision, which is not always possible. We highlight that it is not necessary to match the full-precision performance in order to outperform the MMSE solution, which can be achieved with as few as 6 bits of precision for the smallest MIMO systems. 

Finally, we recommend the best scheme to use for a few specific examples. In Hopfield optimization solvers based on memristors~\cite{Hizzani2024}, a memristor crossbar array is used to calculate the gradient of the objective function to be minimized~\cite{Bhattacharya2024}. Using this type of architecture, also described in Refs.~\cite{Mahmoodi2019} and \cite{Cai2020}, the linear coefficients do not need to be stored in the crossbar array and can instead be applied by a peripheral circuit, so they can have greater precision than quadratic coefficients. In this context, employing a heterogeneous quantization scheme is feasible, and we recommended doing so.

In a CPU/GPU environment as described in Refs.~\cite{Kim2021} and~\cite{Singh2025}, there is complete flexibility in setting the precision of the QUBO coefficients. Even if precision is not an issue with this type of hardware, decreasing it will free up resources and speed up computing, improving the throughput of MIMO detection. In this context, we also recommend using a heterogeneous quantization scheme.

Not all CIMs are designed to solve Ising problems that contain linear terms. This issue can be addressed by adding an auxiliary variable and solving a QUBO problem that includes only quadratic terms~\cite{Singh2022}. In this context, as there will be more quadratic terms than in the original problem, and their distributions will be different from the ones presented in this paper; thus, it is safer to use a homogeneous quantization scheme.

\begin{figure*}[t]
    \centering
    \includegraphics[width=\linewidth]{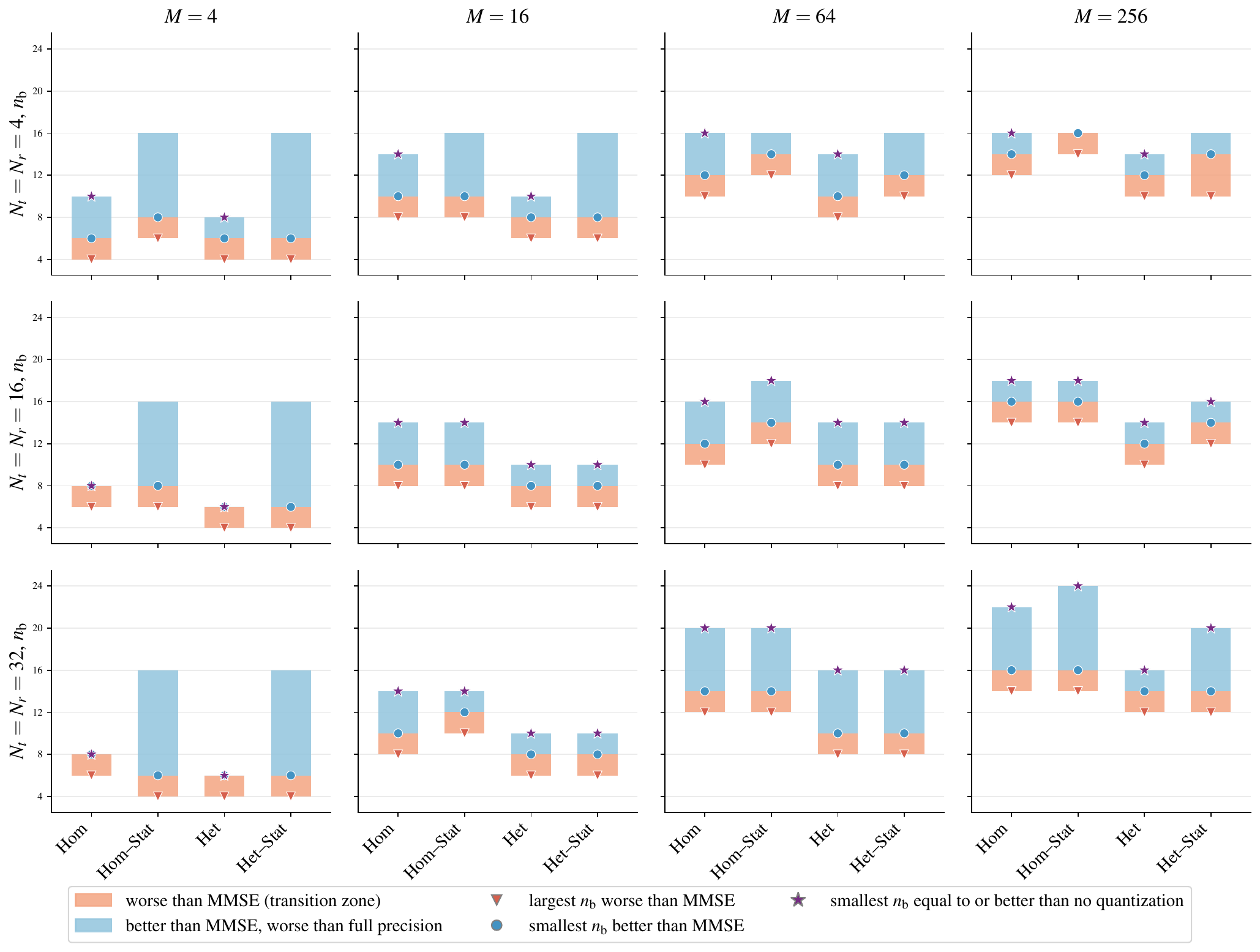}
    \caption{Summary of quantization schemes' performance to highlight important values of $n_{\rm b}$, specifically, the largest value associated with a bit error \mbox{rate (BER)} that is worse than MMSE (orange triangle); the smallest value associated with a BER better than MMSE (blue circle); and the smallest value associated with a BER equal to, or better than, full precision (purple star).} 
    \label{fig:n_bits_all}
\end{figure*}

\begin{figure}[h]
    \centering
    \includegraphics[width=\linewidth]{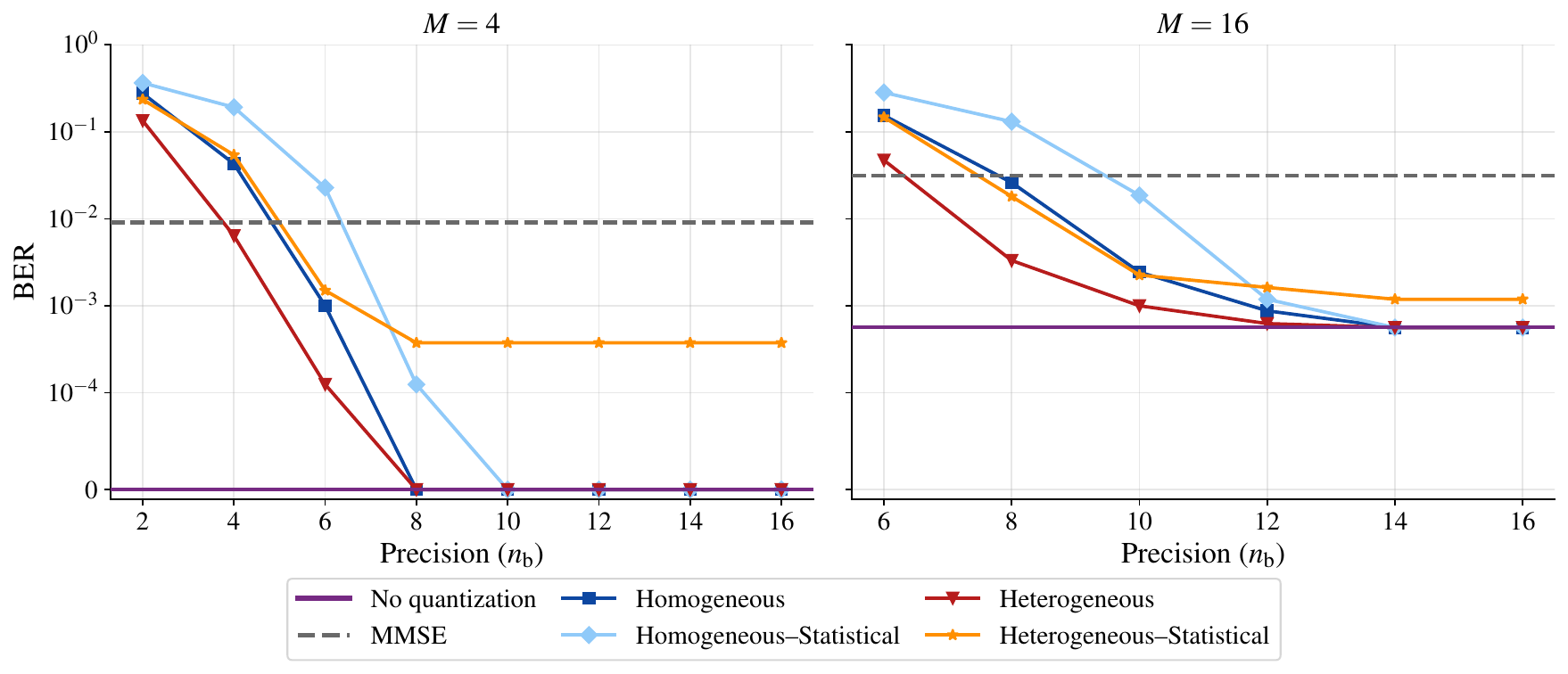}
    \caption{Bit error rate obtained by an exhaustive solver on MIMO systems with $N_t=4$, $N_r=4$, and $E_{\rm b}/N_0=10$~dB.}
    \label{fig:ber_exhaustive}
\end{figure}

\section*{Acknowledgements}
The authors thank Dima Strukov for discussions on hardware architecture designs for QUBO solvers, Thomas Van Vaerenbergh for his valuable input, and Marko Bucyk for his careful review and editing of the manuscript. We thank Hewlett Packard Enterprise for access to their computational resources. This work was supported by DARPA through Air Force Research Laboratory Agreement No. FA8650-23-3-7313. The views expressed are those of the author(s) and do not represent the official views of the Department of Defense or the U.S. Government. S.~H. acknowledges the support of Mitacs.

\bibliographystyle{IEEEtran}
\bibliography{mimo}

\appendices

\section{}\label{app:diagonal}

\emph{Proof of Theorem~\ref{theorem:diagonal}:} Expression~\eqref{eq:QUBO_matrix_entries} gives us the diagonal entries of the QUBO matrix ${\textbf{Q}}$, when $i = j$:
\begin{align}
\label{QUBO_diagonal_entries}
\nonumber
Q_{i,i} &= \sum_{k=1}^{N_r}c_k^*W_{k,i} + c_kW_{k,i}^*+ |W_{k,i}|^2 \\ \nonumber
&=2 \sum_{k=1}^{N_r} \Re(c_k^*W_{k,i}) + |W_{k,i}|^2 \\
&=2 \, \Re \left[\sum_{k=1}^{N_r} W_{k,i}\left(c_k^* + \frac{W^*_{k,i}}{2} \right)\right],
\end{align}
where $\Re(\cdot)$ denotes the real part. From \cref{eq:MIMO_to_QUBO_parameters}, we have
\begin{align}
\label{Bias_term}
\nonumber
c_k&=y_k-\frac{d_{\min}}{2}\cdot(-(\sqrt{M}-1))\cdot(1+\textbf{j})\cdot\sum_{l=1}^{N_t}H_{k,l}\\ \nonumber
&=n_k+\sum_{l=1}^{N_t}x_lH_{k,l}-\frac{d_{\min}}{2} (-(\sqrt{M}-1))(1+\textbf{j})\cdot\sum_{l=1}^{N_t}H_{k,l} \\ \nonumber
&=n_k+\sum_{l=1}^{N_t}(x_l-\frac{d_{\min}}{2}\cdot(-(\sqrt{M}-1))\cdot(1+\textbf{j}))\cdot H_{k,l}\\ 
&=n_k+\sum_{l=1}^{N_t}x^\prime_l\cdot H_{k,l}.
\end{align}
Plugging the right-hand side of \cref{Bias_term} into \cref{QUBO_diagonal_entries}\footnote{$\,$We ignore the $2\,\Re(\cdot)$ portion of \cref{QUBO_diagonal_entries} for now, and consider it later.} yields 
\begin{align}
\label{diag_three_terms}
\nonumber
&\sum_{k=1}^{N_r}\left(c_k^*+\frac{W^*_{k,i}}{2}\right)\cdot W_{k,i}\\ \nonumber &=\sum_{k=1}^{N_r}\left[\left(n_k+\sum_{l=1}^{N_t}x^\prime_l\cdot H_{k,l}\right)^{\!\!*}+\frac{W^*_{k,i}}{2}\right]\cdot W_{k,i}\\ \nonumber
&=\sum_{k=1}^{N_r}\left(\sum_{l=1}^{N_t}(x^\prime_l)^*\cdot H^*_{k,l}+\frac{n_k^*}{N_t}+\frac{W^*_{k,i}}{2N_t}\right)\cdot W_{k,i}\\ \nonumber
&=\sum_{k=1}^{N_r}\sum_{l=1}^{N_t}(x^\prime_l)^*\cdot H^*_{k,l}W_{k,i}+\frac{n^*_k W_{k,i}}{N_t}+\frac{W^*_{k,i} W_{k,i}}{2N_t} \\ \nonumber
&=\sum_{l=1}^{N_t}\sum_{k=1}^{N_r}(x^\prime_l)^*\cdot H^*_{k,l}W_{k,i}+\frac{n^*_k W_{k,i}}{N_t}+\frac{W^*_{k,i}\cdot W_{k,i}}{2N_t} \\ \nonumber
&=\sum_{l=1}^{N_t}(x^\prime_l)^* \sum_{k=1}^{N_r}H^*_{k,l}W_{k,i} \! + \! N_t\sum_{k=1}^{N_r}\frac{n^*_k W_{k,i}}{N_t} \! + \! N_t\sum_{k=1}^{N_r}\frac{|W_{k,i}|^2}{2N_t}\\ 
&=\!\sum_{l=1}^{N_t}(x^\prime_l)^*\! \sum_{k=1}^{N_r}H^*_{k,l} W_{k,i}+\sum_{k=1}^{N_r}n^*_k W_{k,i}\! +\!\frac{1}{2}\sum_{k=1}^{N_r}|W_{k,i}|^2\!.
\end{align}

The right-hand side of \cref{diag_three_terms} consists of three terms. The first term can be written as
\begin{align}
\label{diag_first_term}
\nonumber
&\!\!\!\!\sum_{l=1}^{N_t}(x^\prime_l)^*\cdot\sum_{k=1}^{N_r}H^*_{k,l}W_{k,i}\\ \nonumber
&\!\!=\sum_{l=1}^{N_t}(x^\prime_l)^*\cdot\sum_{k=1}^{N_r}H^*_{k,l}(-H_{k,\lfloor{i/r}\rfloor}\cdot e^{\,\textbf{j}\frac{\pi}{2}\lfloor\frac{i \!\!\!\! \mod \!r}{r/2}\rfloor} \\ \nonumber
& \ \ \cdot 2^{[(i \!\!\!\!\mod
\!r)\!\!\!\!\mod \!\frac{r}{2}]+1}\cdot\frac{d_{\min}}{2})\\ \nonumber
&\!\!=(-e^{\,\textbf{j}\frac{\pi}{2}\lfloor\frac{i\!\!\!\!\mod \!r}{r/2} \rfloor} \cdot 2^{[(i\!\!\!\!\mod\! r)\!\!\!\!\mod \!\frac{r}{2}]+1} \cdot \frac{d_{\min}}{2})\\ \nonumber
& \ \ \cdot \sum_{l=1}^{N_t}(x^\prime_l)^*\cdot\sum_{k=1}^{N_r}H^*_{k,l}\cdot H_{k,\lfloor{i/r} \rfloor}\\ \nonumber
&\!\!=a_i \sum_{l=1}^{N_t}(x^\prime_l)^*\cdot\sum_{k=1}^{N_r}H^*_{k,l}\cdot H_{k,\lfloor{i/r} \rfloor}\\ 
&\!\!=\!a_i \!\!\left[ \! \! \left(\sum_{l=1,l\ne \lfloor{i/r\!} \rfloor}^{N_t}\!\! \! \! (x^\prime_l)^* \!\sum_{k=1}^{N_r}\!H^*_{k,l} H_{k,\!\lfloor{\!i/r\!} \rfloor} \!\! \right) \!\! + \! (x^\prime_{\!\lfloor{\!i/r\!} \rfloor\!})^*\! \sum_{k=1}^{N_r}\!|H_{k,\!\lfloor{\!i/r}\! \rfloor}\!|^2\!\right]\!\!,
\end{align}
where 
\begin{equation} \nonumber
    a_i = -e^{\,\textbf{j}\frac{\pi}{2}\lfloor\frac{i\!\!\!\!\mod\!r}{r/2}\rfloor} \cdot 2^{[(i\!\!\!\! \mod\!r)\!\!\!\!\mod\!\frac{r}{2}]+1} \cdot \frac{d_{\min}}{2}.
\end{equation}

The third term of the right-hand side of \cref{diag_three_terms} can be written as 
\begin{align} \nonumber
\label{diag_third_term}
\frac{1}{2}\!\sum_{k=1}^{N_r}\!|W_{k,i}|^2 & = \frac{1}{2}(2^{[(i\!\!\!\!\mod\!r)\!\!\!\!\mod\! \frac{r}{2}]+1}\!\cdot\!\frac{d_{\min}}{2})^2\!\sum_{k=1}^{N_r}\!|H_{k,\lfloor{i/r} \rfloor}\!|^2 \\
&= \frac{|a_i|^2}{2}\sum_{k=1}^{N_r}|H_{k,\lfloor{i/r} \rfloor}\!|^2.
\end{align}
Plugging the right-hand sides of \cref{diag_first_term} and \cref{diag_third_term} into \cref{diag_three_terms} yields the expression
\begin{align}
\label{diag_three_terms_rev} \nonumber
& a_i \left(\sum_{l=1,l\ne \lfloor{i/r} \rfloor}^{N_t}(x^\prime_l)^*\cdot\sum_{k=1}^{N_r}H^*_{k,l}\cdot H_{k,\lfloor{i/r} \rfloor}\right) +\sum_{k=1}^{N_r}n^*_k W_{k,i} \\ 
& \ \  + \left[a_i \cdot (x^\prime_{\lfloor{i/r} \rfloor})^*+ \frac{|a_i|^2}{2}\right] \cdot \left( \sum_{k=1}^{N_r}|H_{k,\lfloor{i/r} \rfloor}|^2\right)\!\!.
\end{align}
To find the probability distribution function of expression~\eqref{diag_three_terms_rev}, we first find the mean and the variance of each of its three terms, and then apply the central limit theorem (CLT). For the first term, let us start with its mean:
\begin{align} \nonumber
\label{first_term_mean}
&\mathbb{E}\left(\sum_{l=1,l\ne \lfloor{i/r} \rfloor}^{N_t}(x^\prime_l)^*\cdot\sum_{k=1}^{N_r}H^*_{k,l}\cdot H_{k,\lfloor{i/r} \rfloor}\right) \\
 &\ \ = \! \sum_{l=1,l\ne \lfloor{i/r} \rfloor}^{N_t}\sum_{k=1}^{N_r} \mathbb{E}((x^\prime_l)^*) \cdot \mathbb{E}(H^*_{k,l})\cdot \mathbb{E}(H_{k,\lfloor{i/r} \rfloor})=0,
\end{align}
because $H^*_{k,l}$ and $H_{k,\lfloor{i/r} \rfloor}$ are independent and have a mean of zero.
For the variance of the first term, note that for any values of $k$ and $l$, 
\begin{equation}
\mathbb{E}(|H^*_{k,l} H_{k,\lfloor{i/r} \rfloor}|^2)
=\mathbb{E}(|H^*_{k,l}|^2)\cdot \mathbb{E}(|H_{k,\lfloor{i/r} \rfloor}|^2)=1 
\end{equation}
and 
\begin{equation}
\mathbb{E}(|x^\prime_l|^2) = 2\cdot \frac{2\sqrt{M}-1}{\sqrt{M}+1}.
\end{equation}
Thus, 
\begin{equation} \label{eq: first term var}
\mathrm{Var}\left((x^\prime_l)^*\cdot\sum_{k=1}^{N_r}H^*_{k,l}\cdot H_{k,\lfloor{i/r} \rfloor}\right)= 2\cdot \frac{2\sqrt{M}-1}{\sqrt{M}+1}\cdot N_r.
\end{equation}
For the second term in expression~\eqref{diag_three_terms_rev}, we have
\begin{equation} \label{eq: second term mean}
\mathbb{E}(n^*_kW_{k,i})=\mathbb{E}(n^*_k)\cdot\mathbb{E}(W_{k,i})=0,
\end{equation}
and its variance is 
\begin{equation} \label{eq: second term var}
\mathrm{Var}(n^*_kW_{k,i}) =\mathbb{E}(|n^*_k|^2)\cdot \mathbb{E}(|W_{k,i}|^2)=\sigma^2_n\cdot|a_i|^2,
\end{equation}
where $\sigma^2_n$ is the variance of the noise. The third term in expression~\eqref{diag_three_terms_rev} can be written as
\begin{equation}
\sum_{k=1}^{N_r}(a_i \cdot (x^\prime_{\lfloor{i/r} \rfloor})^*+ \frac{|a_i|^2}{2})\cdot|H_{k,\lfloor{i/r} \rfloor}|^2.
\end{equation}
Now, for any $i$ and $k$, since $\mathbb{E}(|H_{k,\lfloor{i/r} \rfloor}|^2) = 1$, the expected value of the third term is
\begin{align} \label{eq: third term mean}
\nonumber
g_i & = \mathbb{E}\left(\left[a_i \cdot (x^\prime_{\lfloor{i/r} \rfloor})^*+ \frac{|a_i|^2}{2}\right]\cdot|H_{k,\lfloor{i/r} \rfloor}|^2\right)\\ \nonumber
&=\mathbb{E}\left(\left[{a_i \cdot (x^\prime_{\lfloor{i/r} \rfloor})^*+ \frac{|a_i|^2}{2}}\right]\right)\cdot\mathbb{E}(|H_{k,\lfloor{i/r} \rfloor}|^2)\\ 
&=\!\left(a_i\left[\sqrt{\frac{6}{M-1}}\cdot\frac{(\sqrt{M}-1)}{2} \cdot (1-\textbf{j})\right] +\frac{|a_i|^2}{2}\right)\!\!.
\end{align}
We calculate the variance of the third term in expression~\eqref{diag_three_terms_rev}:
\begin{align} \label{eq: var third term} \nonumber
& \mathrm{Var}\left(\left[a_i \cdot (x^\prime_{\lfloor{i/r} \rfloor})^*+ \frac{|a_i|^2}{2}\right]\cdot|H_{k,\lfloor{i/r} \rfloor}|^2\right)\\ 
& \ \ =\mathbb{E}\left(\left|\left[a_i \cdot (x^\prime_{\lfloor{i/r} \rfloor})^*+ \frac{|a_i|^2}{2}\right]\cdot|H_{k,\lfloor{i/r} \rfloor}|^2\right|^2\right)-|g_i|^2. 
\end{align}
The first term of \cref{eq: var third term} is equal to
\begin{align}
\nonumber
f_i & = \mathbb{E}\left(\left|\left[a_i \cdot (x^\prime_{\lfloor{i/r} \rfloor})^*+ \frac{|a_i|^2}{2}\right]\right|^2\right)\cdot  \mathbb{E}(|H_{k,\lfloor{i/r} \rfloor}|^4) \\ \nonumber
& = \left(|a_i|^2 \cdot 2\cdot \frac{2\sqrt{M}-1}{\sqrt{M}+1} + \frac{|a_i|^4}{4} \right. \\  
&\,\,\left.\,+\, 2\,\Re\!\left[a_i\cdot\frac{|a_i|^2}{2}\cdot  \sqrt{\frac{6}{M-1}}\cdot\frac{(\sqrt{M}-1)}{2} \cdot (1-\textbf{j}) \right]\right) \cdot 2.
\vspace{-0.5em}
\end{align}
\vspace{-0.5em}
\noindent Thus,
 \begin{equation} \label{eq: third term var}
\mathrm{Var}\left(\left[a_i \cdot (x^\prime_{\lfloor{i/r} \rfloor})^*+ \frac{|a_i|^2}{2}\right]\cdot|H_{k,\lfloor{i/r} \rfloor}|^2\right) = f_i - |g_i|^2.
\end{equation}
By using the mean and the variance of expression~\eqref{diag_three_terms_rev}'s terms from Eqs.~\eqref{first_term_mean}, \eqref{eq: first term var}, \eqref{eq: second term mean}, \eqref{eq: second term var}, \eqref{eq: third term mean}, and~\eqref{eq: third term var}, and applying the CLT, the PDF of the sum~\eqref{diag_three_terms_rev} can be estimated by
\begin{align}
\label{final_res_diag}
\nonumber
& \mathcal{CN}[(N_r \cdot g_i), \ \   (|a_i|^2\cdot (N_t-1) \cdot 2\cdot \frac{2\sqrt{M}-1}{\sqrt{M}+1}\cdot N_r \\ 
&\ \ + N_r\cdot \sigma^2_n\cdot|a_i|^2 + N_r (f_i - |g_i|^2))].
\end{align}
Finally, if we apply 
$2\,\Re(\cdot)$
to distribution~\eqref{final_res_diag}, we have
\begin{equation}
\begin{aligned}
\label{diag_term_final}
Q_{i,i}\sim &\ \mathcal{N}[\left(2\cdot N_r \cdot \Re(g_i)\right), \\ 
& \ \ (2 (|a_i|^2\cdot (N_t-1) \cdot 2\cdot \frac{2\sqrt{M}-1}{\sqrt{M}+1}\cdot N_r\\ 
& \ \ \ \ + N_r\cdot \sigma^2_n\cdot|a_i|^2 + N_r (f_i - |g_i|^2))].
\end{aligned}
\end{equation}
\section{}\label{app:off-diagonal}
\subsection{Proof of Theorem~\ref{theorem:off-diagonal}: -- Case 1}
\noindent From expression~\eqref{eq:QUBO_matrix_entries}, for $i \neq j$, we have
\begin{align}
\label{QUBO_non_diag} \nonumber
Q_{i,j}&=\sum_{k=1}^{N_r}W_{k,i}^*W_{k,j}+W_{k,j}^*W_{k,i}=2\sum_{k=1}^{N_r}\Re(W_{k,i}^*W_{k,j})\\
&=2\,\Re\left(\sum_{k=1}^{N_r}W_{k,i}^*W_{k,j}\right)\!\!.
\end{align}
Since $H_{k,\lfloor i/r \rfloor} \sim \mathcal{CN}[0,\, 1]$, and using \cref{eq:MIMO_to_QUBO_parameters_2}, we have
\begin{align}
\nonumber \label{W_distribution}
W_{k,i}&=-H_{k,\lfloor i/r \rfloor}\cdot e^{\,\textbf{j}\frac{\pi}{2}(\frac{i\!\!\!\!\mod \!r}{r/2})}\cdot 2^{[(i\!\!\!\!\mod \!r)\!\!\!\!\mod
\!\frac{r}{2}]+1}\cdot \frac{d_{\min}}{2}\\
& \sim \mathcal{CN}[0,\ |a_i|^2].
\end{align}
$W_{k,i}$ and $W_{k,j}$ are functions of $H_{k,\lfloor i/r \rfloor}$ and $H_{k,\lfloor j/r \rfloor}$, respectively. If $\lfloor i/r \rfloor\ne \lfloor j/r \rfloor$, $W_{k,i}$ and $W_{k,j}$ correspond to different entries of $\mathbf{H}$ and are independent. Thus, we have
\begin{align}
\mathbb{E}(W_{k,i}^*W_{k,j}) = \mathbb{E}(W_{k,i}^*)\mathbb{E}(W_{k,j})=0
\end{align}
and
\begin{align}
\mathrm{Var}(W_{k,i}^*W_{k,j})= \mathbb{E}(|W_{k,i}^*|^2)\mathbb{E}(|W_{k,j}|^2)= |a_i|^2\cdot|a_j|^2. 
\end{align}
Now, by applying the CLT to \cref{QUBO_non_diag}, for $i\ne j$, with \mbox{$\lfloor i/r \rfloor\ne \lfloor j/r \rfloor$,} we arrive at
%
\begin{equation}
\label{QUBO_non_diag_case1_result}
Q_{i,j} = \Re\left(\sum_{k=1}^{N_r}W_{k,i}^*W_{k,j}\right) \sim \mathcal{N}[0,\  \left(2\cdot N_r\cdot|a_i|^2\cdot|a_j|^2\right)].
\end{equation}

\subsection{Proof of Theorem~\ref{theorem:off-diagonal} -- Case 2}
\noindent Taking into account $i=j$ and $\lfloor i/r \rfloor\ne \lfloor j/r \rfloor$, we have
\begin{align}
\label{QUBO_non_diag_case2} \nonumber
&W_{k,i}^*W_{k,j} = |H_{k,\lfloor i/r \rfloor}|^2\cdot e^{\,\textbf{j}\frac{\pi}{2}(\frac{j\!\!\!\!\mod\! r}{r/2}-\frac{i\!\!\!\!\mod\! r}{r/2})}\\
&\ \ \cdot2^{[(i\!\!\!\!\mod\! r)\!\!\!\!\mod\! \frac{r}{2}]\,+\,[(j\!\!\!\!\mod\! r)\!\!\!\!\mod\! \frac{r}{2}]\,+\,2}\cdot\biggl(\frac{d_{\min}}{2}\biggl)^2.
\end{align}
If $\frac{j\!\!\mod r}{r/2}\ne\frac{i\!\!\mod r}{r/2}$, $W_{k,i}^*W_{k,j}$ will always be imaginary. Thus, when we apply $2\,\Re(\cdot)$, the result will always be 0. 

\subsection{Proof of Theorem~\ref{theorem:off-diagonal} -- Case 3}
\noindent Similarly to case 2, consider $\lfloor i/r \rfloor=\lfloor j/r \rfloor$, but this time assume $\frac{j\!\!\mod r}{r/2}=\frac{i\!\!\mod r}{r/2}$. We then have
\begin{align}
\label{QUBO_non_diag_case3}
\nonumber
W_{k,i}^*W_{k,j}  = &   \,\,|H_{k,\lfloor i/r \rfloor}|^2 \\ \nonumber
& \cdot2^{[(i\!\!\!\!\mod\! r)\!\!\!\!\mod\!\frac{r}{2}]\,+\,[(j\!\!\!\!\mod\! r)\!\!\!\! \mod\!\frac{r}{2}]\,+\,2}\!\cdot\!\biggl(\frac{d_{\min}}{2}\biggl)^{\!2}\\ \nonumber
= & \ |H_{k,\lfloor i/r \rfloor}|^2\cdot |a_i|\cdot|a_j|.
\end{align}
Now, after applying the CLT, we have
\begin{equation}
\label{QUBO_non_diag_case3_result}
Q_{i,j}\sim \mathcal{N}[(2\cdot N_r\cdot|a_i|\cdot|a_j|),\ (4\cdot N_r\cdot|a_i|^2\cdot|a_j|^2)].
\end{equation}

\end{document}